\documentclass{llncs}

\usepackage{amsmath, amssymb, mathrsfs, bm, latexsym,graphics,color, graphicx,url,floatflt}

\begin{document}

\title{An exact algorithm for the weighed mutually exclusive maximum set cover problem}

\author{Songjian Lu and Xinghua Lu}

\institute{Department of Biomedical Informatics,\\ University of
Pittsburgh, Pittsburgh, PA 15219, USA\\ Email: songjian@pitt.edu,
xinghua@pitt.edu}

\maketitle

\begin{abstract}
In this paper, we  introduce an exact algorithm with a time
complexity of $O^*(1.325^m)^{\dag}$
\let\thefootnote\relax\footnotetext{$^{\dag}${\bf Note:} Following
the recent convention, we use a star $*$ to represent that the
polynomial part of the time complexity is neglected.} for the {\sc
weighted mutually exclusive maximum set cover} problem,  where $m$
is the number of subsets in the problem. This  is an  NP-hard 
motivated and abstracted from a bioinformatics problem of identifying signaling
pathways based gene mutations.   Currently, this probelm is addressed using 
heuristic algorithms, which cannot guarantee the performance of
the solution.  By providing a relatively efficient exact algorithm, 
our approach will like increase the capability of finding better
solutions in the application of cancer research.
\end{abstract}

\section{Introduction}
Cancers are genomic diseases in that genomic perturbations, such as mutation
of genes, lead to perturbed cellular signal pathways, which in turn
lead to uncontrolled cell growth.  An important cancer research area 
is to discover perturbed signal transduction pathways in cancers, in 
order to gain insights in disease mechanisms and guide patient treatment.
It has observed that mutation events among the genes constitute a signaling pathway
tend to be occur in a mutually exclusive fashion \cite{sparks1998,TCGA_gbm}.  This
is because often one mutation in such a pathway may be 
usually sufficient to disrupt the signal carried by a pathway leading to cancers.
Contemporary biotechnologies can easily detect what genes have mutated in
tumor cells, providing an unprecedented opportunity to study cancer signaling  pathways. 
However, as each tumor usually has up to hundreds
of  mutations, some dispersed in different pathways driving tumor genesis while others mutations
not related to cancers, it is a challenge to find mutations across different patients that affect a  common 
cancer signaling pathway.    The property of mutual exclusivity of mutations in a common
pathway can help us to recognize  driver mutations within a common pathway~\cite{Ciriello,Miller,Vandin2012}. 

The problem of finding mutations within a common pathway across tumors, i.e., finding the members of the pathway,   can be cast as follows:
finding a set of mutually exclusive mutations that cover a maximum
number of tumors.  This is an NP-hard (this problem is abstracted to the
{\sc mutually exclusive maximum set cover} problem), and previous studies~\cite{Ciriello,Miller,Vandin2012} used heuristic algorithms to solve the problem, which could not guarantee 
the optimal solutions. Another shortcoming of the previous studies is that they do not 
consider the weight of the mutations.  Since the signal carried by a signaling pathway is 
often reflected as a phenotype, statistical methods can be used to assign a weight to a type of gene mutation by assessing the strength of association of the mutation event and appearance of a phenotype of interest.  Therefore,
it is more biologically interesting to find a set of mutually exclusive mutations that carries 
as much weight as possible and covers as many tumors as possible---thus
a {\sc weighted mutually exclusive maximum set cover} problem.  

{\sc mutually exclusive maximum set cover} problem is: given a
ground set $X$ of $n$ elements, a collection ${\cal F}$ of $m$
subsets of $X$, try to find a sub-collection ${\cal F'}$ of ${\cal
F}$ with minimum number of subsets such that 1) no two subsets in
${\cal F'}$ are overlapped and 2) ${\cal F'}$ covers the maximum
number of elements in $X$, i.e. the number of elements of the
union of all subsets in ${\cal F'}$ is maximized. If we assign
each subset in ${\cal F}$ a weight (a real number) and further
require that the weight of ${\cal F'}$, i.e. the weight sum of
subsets in ${\cal F'}$, is minimized, then the {\sc mutually
exclusive maximum set cover} problem becomes the {\sc weighted
mutually exclusive maximum set cover} problem.

The research on the {\sc mutually exclusive maximum set cover} and
the {\sc weighted mutually exclusive maximum set cover} problems
is limit. To our best knowledge, only Bj\"olund et al.
~\cite{bjorklund} gave an algorithm of $O^*(2^n)$ for the problem
of finding $k$ subsets in ${\cal F}$ with maximum weight sum that
cover all elements in $X$ (the solution may not exists). The {\sc
mutually exclusive maximum set cover} problem is obtained by
adding constrains to the {\sc set cover} problem, which is a
well-known NP-hard problem in Karp's 21 NP-complete
problems~\cite{Karp1972}. Much research about the {\sc set cover}
problem has been focused on the approximation algorithms, such as
papers ~\cite{alon,feige,Kolliopoulos,lund} gave polynomial time
approximation algorithms that find solutions whose sizes are at
most $c\log n$ times the size of the optimal solution, where $c$
is a constant. There is also plenty research about the {\sc
hitting set} problem, which is equivalent to the {\sc set cover}
problem. In this direction, people mainly designed fixed-parameter
tractable (FPT) algorithms that used the solution sizes $k$ as
parameter for the {\sc hitting set} problem under the constrain
that sizes of all subsets in the problem are bounded by $d$. For
example, Niedermeier et al. \cite{niedemedier} gave a
$O^*(2.270^k)$ algorithm for the {\sc $3$-hitting set} problem,
and Fernau et al.~\cite{fernau_2} gave a $O^*(2.179^k)$ algorithm
respectively. Very recently, people also studied the extension
version of the {\sc set cover} problem that find a sub-set ${\cal
F'}$ of ${\cal F}$ such that each element in $X$ is covered by at
least $t$ subsets in ${\cal F'}$. For example, Hau et
al.~\cite{Hua2} designed an algorithm with time complexities of
$O^*((t+1)^n)$ for the problem; Lu et al.~\cite{Lu2011} further
improved the algorithm under the constrain that there are certain
elements in $X$ are included in at most $d$ subsets in ${\cal F}$.
These two algorithms can be easily modified to solved the {\sc
weighted mutually exclusive maximum set cover} problem. However,
as in application, $n$, the number of tumor samples, is large (can
be several hundreds). Above two algorithms are not practical. On
the other hand, by excluding somatic mutations that are less
possible to be related to a pathway in the study, the number of
mutations is usually less than the number of tumors. Hence, there
is a need to design better algorithms solving the {\sc weighted
mutually exclusive maximum set cover} problem and using $m$ the
number of subsets (mutations) in ${\cal F}$ as parameter.

In this paper, first, we will prove that the {\sc weighted
mutually exclusive maximum set cover} problem is NP-hard. Then, we
will give an algorithm of running time bounded by $O^*(1.325^m)$
for the problem. This running time complexity is only the worst
case upper bound. In our test, this algorithm could solve the
problem practically when it was applied to the TCGA
data~\cite{TCGA_gbm} for searching the diver mutations.

\section{The {\sc weighted mutually exclusive maximum set cover} problem is NP-hard}
The formal definition of the {\sc weighted mutually exclusive
maximum set cover} problem is: given a ground set $X$ of $n$
elements, a collection ${\cal F}$ of $m$ subsets of $X$, and a
weight function $w: {\cal F} \rightarrow [0, \infty)$, if ${\cal
F'} =\{S_1,S_2,\ldots,S_h\} \subset {\cal F}$ such that
$|(\cup_{i=1}^hS_i)|$ is maximized, and $S_i\cap S_j=\emptyset$
for any $i \neq j$, then we say ${\cal F'}$ is a mutually
exclusive maximum set cover of $X$ and $\sum_{i=1}^hw(S_i)$ is the
weight of ${\cal F'}$; the goal of the problem is to find a
mutually exclusive maximum set cover of $X$ with the minimum
weight.

In this section, we will prove that the {\sc mutually exclusive
maximum set cover} problem, i.e. all subsets in ${\cal F}$ have
equal weight, is NP-hard, which would in turn prove that the {\sc
weighted mutually exclusive maximum set cover} problem is NP-hard.

We will prove the NP-hardness of the {\sc mutually exclusive
maximum set cover} problem by reducing another NP-hard problem,
the {\sc maximum $3$-set packing} problem, to it. Recall that the
{\sc maximum $3$-set packing} problem is: given a collection
${\cal F}$ of subsets, where the size of each subset in ${\cal F}$
is $3$, try to find an ${\cal S} \subset {\cal F}$ such that
subsets in ${\cal S}$ are pairwise disjoint and $|{\cal S}|$ is
maximized.

\begin{theorem}\label{theorem1}
The {\sc mutually exclusive maximum set cover} problem is NP-hard.
\end{theorem}
\begin{proof}
Let ${\cal S} = \{S_1,S_2,\ldots,S_m\}$ be an instance of the {\sc
maximum $3$-set packing} problem. We create an instance of the
{\sc mutually exclusive maximum set cover} problem such that $X =
\cup_{i=1}^mS_i$ and ${\cal F} = {\cal S}$.

It is obvious that ${\cal P} = \{P_1,P_2,\ldots, P_k\}$ is a
solution of the {\sc mutually exclusive maximum set cover} problem
if and only if ${\cal P} = \{P_1,P_2,\ldots, P_k\}$ is a solution
of the {\sc maximum $3$-set packing} problem. Thus, the {\sc
mutually exclusive maximum set cover} problem is NP-hard. \qed
\end{proof}

\section{The main Algorithm}

In this section, we will introduce our main algorithm. The basic
idea of our method is branch and bound. The algorithm first finds
a subset in ${\cal F}$ and then branches on it. By the mutual
exclusivity, if any two subsets in ${\cal F}$ are overlapped, then
at most one of them can be chosen into the solution. Hence,
suppose that the subset $S$ intersects with other $d$ subsets in
${\cal F}$, then if $S$ is included into the solution, $S$ and
other $d$ subsets intersected with $S$ will be removed from the
problem, and if $S$ is excluded from the solution, $S$ will be
removed from the problem. We continue this process until the
resulting sub-problems can be solved in constant or polynomial
time.

The execution process of the algorithm is going through a search
tree and the running of the algorithm is proportional to the
number of leaves in the search tree. If letting $T(m)$ be the
number of leaves of search tree when call the algorithm with $m$
subsets in ${\cal F}$, then we can obtain the recurrence relation
$T(m) \leq T(m -(d+1)) + T(m-1)$. As if $d=0$, the problem can be
solved in polynomial time (all subsets in ${\cal F}$ will be
included into the solution), $d\geq 1$. Therefore, we can obtain
$T(m) \leq 1.619^m$, which means the problem can be solved in
$O^*(1.619^m)$ time$^{\ddag}$
\let\thefootnote\relax\footnotetext{$^{\ddag}${\bf Note:} Given a
recurrence relation $T(k) \leq \sum_{i=0}^{k-1}c_iT(i)$ such that
all $c_i$ are nonnegative real numbers, $\sum_{i=0}^{k-1}c_i>0$,
and $T(0)$ represents the leaves, then $T(k) \leq r^k$, where $r$
is the unique positive root of the characteristic equation $t^k -
\sum_{i=0}^{k-1}c_it^i=0$ deduced from the recurrence
relation~\cite{chen}.}. In this paper, we will improve the running
time to solve the problem by carefully selecting subsets in ${\cal
F}$ for branching.

Before present our major result, we prove three lemmas. Given an
instance $(X, {\cal F}, w)$ of the {\sc weighted mutually
exclusive maximum set cover} problem, we make a graph $G$ called
the set interaction graph such that each subset in ${\cal F}$
makes a node in $G$ and if any two subsets are interacted, an edge
is added between them.

For the convenience, in the rest of paper, we will use a node in
the intersection graph and a subset in ${\cal F}$ in a mixed way.
Suppose $C=(V_c, E_c)$ is a connected component of $G$, we denote
$(X,{\cal F},w)_C$ the sub-instance induced by component $C$, i.e.
$(X,{\cal F},w)_C = (\cup_{S\in V_c}S,V_c,w)$. In the algorithm,
when we say $Solution_1$ is better than $Solution_2$ if 1)
$Solution_1$ covers more elements in $X$ than $Solution_2$ covers,
or 2) $Solution_1$ and $Solution_2$ cover the same number of
element in $x$, however the weight of $Solution_1$ is less than
the weight of $Solution_2$. In the intersection graph,
$neighbor(S)$ includes $S$ and all nodes that are connected to
$S$.

The first lemma will show that we can find the solution of the
problem by finding the solutions of all sub-instances induced by
connected components of the intersection graph $G$.

\begin{lemma}\label{lemma1}
Given an instance $(X, {\cal F}, w)$ of the {\sc weighted mutually
exclusive maximum set cover} problem, if the intersection graph
obtained from the instance consists of several connected
components, then the solution of the problem is the union of
solutions of all sub-instances induced by connected components.
\end{lemma}
\begin{proof}
As the subset(s) in each sub-instance has(have) no element(s) in
other sub-instance(s), we can solve each sub-instance
independently. It obvious that the optimal solutions of all
sub-instances will make the optimal solution of the original
instance. \qed
\end{proof}

In next lemma, we will show that if the maximum degree of the
intersection graph obtained from the given instances is bounded by
$2$, i.e. each subset in the instance is overlapped with at most
other $2$ subsets, then the problem can be solved in polynomial
time.

\begin{figure}[htb]
%\vspace{-3mm}%\setbox4=\vbox{\hsize32pc
%\noindent\strut \begin{quote}
\footnotesize
\begin{tabbing}
xxxx\=xx\=xx\=xx\=xx\=xx\=xx\=xx\=xx\=xx\=xx\=xx\=xx\=\kill
{\bf WMEM-Cover-2$((X,{\cal F},w))$}\\
\textbf{Input:} An instance of the {\sc weighted mutually exclusive maximum set cover}  \\
\>\> problem such that the degree of the interaction graph is bounded by $2$.\\
\textbf{Output:} A mutually exclusive maximum set cover with minimum weight.\\

1 \> {\bf if} $X=\emptyset$ or ${\cal F} = \emptyset$ {\bf then}\\
1.1 \>\> {\bf return} $\emptyset$;\\
2 \> Find all connected components in the intersection graph and save them in $Comp$;\\
3.\> {\bf if} the number of components is larger than $2$ {\bf then} \\
3.1 \>\>\ {\bf return} $\bigcup_{C \in Comp}$ {\bf WMEM-Cover-2$((X,{\cal F},w)_C)$}; \\
\>\> // $(X,{\cal F},w)_C$ represents the sub-instance induced by component $C$.\\
\\
\> {\bf else}\\
3.2 \>\> find the node $x$ such that $x$ is the middle node if the intersection graph\\
\>\>~~is a path or $x$ is any node if the intersection graph is a ring; \\
3.3 \>\> $Solution_1 = \{x\} \cup$ {\bf WMEM-Cover-2$((X-x,{\cal F}-neighbor(x),w))$};\\
\>\> // The $neighbor(x)$ includes $x$ and all nodes that are
connected to $x$.\\
\\
3.4 \>\> $Solution_2 = $ {\bf WMEM-Cover-2$((X,{\cal F}-x,w))$};\\
3.5 \>\> {\bf return} the best solution among $Solution_1$ and
$Solutoin_2$;\\
\>\> // The best solution either covers more elements in $X$ than
other solutions cover or\\
\>\>~ ~ has minimum weight if all solutions cover the same number
of elements in $X$.
\end{tabbing}
%\end{quote}
\vspace*{-3mm}
%  \strut} $$\boxit{\box4}$$ \vspace*{-15mm}
\caption{Algorithm for the {\sc weighted mutually exclusive
maximum set cover} problem with overlapped degrees bounded by
$2$.} \label{Algorithm_2}
\end{figure}

\begin{lemma}\label{lemma2}
Given an instance $(X, {\cal F}, w)$ of the {\sc weighted mutually
exclusive maximum set cover} problem, if the degree of its
intersection graph is bounded by $2$, then the problem can be
solved in $O(m^2)$ time.
\end{lemma}
\begin{proof}
We first prove that if the intersection graph has only one
connected component, the running time of the algorithm {\bf
WMEM-Cover-2} is polynomial.

As the degree of the intersection graph is bounded by $2$, the
connected component can only be a simple path or a simple ring.

Case 1: Suppose that the intersection graph is a simple path. The
algorithm first finds the middle node (subset) $x$ of the path;
then branches on $x$ such that branch one includes the node into
the solution (three subsets will be removed from the problem) and
branch two excludes the node from the solution (one subset will be
removed from the problem). Hence, if $T(m)$ represents the number
of leaves in the search tree, we will have
\[T(m) \leq T(m-3) +
T(m-1).\] Furthermore, considering that after the branching, the
resulting intersection graphs will be split into two connected
components with almost equal sizes, we have \[T(m)\leq (T(\lceil
(m-3)/2\rceil) + T(\lfloor(m-3)/2\rfloor)) +
(T(\lceil(m-1)/2\rceil) + T(\lfloor(m-1)/2\rfloor)) < 4T(m/2).\]
From this recurrence relation, we will have \[T(m) \leq 4^{\log m}
= m^2.\]

Case 2: Suppose that the intersection graph is a simple ring. The
algorithm chooses any node and branches on it. Similar to case 1,
one branch will remove three subsets from the problem while other
branch will remove one subset from the problem. Hence, we will
have the recurrence relation \[T(m) \leq T(m-3) +  T(m-1).\]
Furthermore, after this operation, the resulting intersection
graphs in both branches are simple pathes. So with the analysis of
case 1, we can obtain \[T(m) \leq (m-3)^2 + (m-1)^2 < 2m^2.\]

If the intersection graph of the instance has multiple connected
components, then by Lemma~\ref{lemma1}, we can solve sub-instances
induced by connected components independently. As each
sub-instance induced by a connected component can be solved in
polynomial time, the original instance can be solved in polynomial
time. It is easy to obtain that the running time is bounded by
$O(m^2)$.

The correctness of the algorithm is straightforward. The algorithm
{\bf WMEM-Cover-2} first chooses a node in the intersection graph,
then branches on it. One branch includes the node into the
solution while the other branch excludes the node from the
solution. Hence, all possible combinations of mutually exclusive
covers are considered and the algorithm will returns the best
solution, i.e. the solution that covers maximum number of elements
in $X$ and has the minimum weight. \qed
\end{proof}

\begin{figure}[htb]
%\vspace{-3mm}%\setbox4=\vbox{\hsize32pc
%\noindent\strut \begin{quote}
\footnotesize
\begin{tabbing}
xxxx\=xx\=xx\=xx\=xx\=xx\=xx\=xx\=xx\=xx\=xx\=xx\=xx\=\kill
{\bf WMEM-Cover-3$((X,{\cal F},w))$}\\
\textbf{Input:} An instance of the {\sc weighted mutually exclusive maximum set cover}  \\
\>\> problem such that the degree of the interaction graph is bounded by $3$.\\
\textbf{Output:} A mutually exclusive maximum set cover with minimum weight.\\
\\

1 \> {\bf if} $X=\emptyset$ or ${\cal F} = \emptyset$ {\bf then}\\
1.1 \>\> {\bf return} $\emptyset$;\\
2 \> Find all connected components in the intersection graph and save them in $Comp$;\\
3.\> {\bf if} the number of components is larger than $2$ {\bf then} \\
3.1 \>\>\ {\bf return} $\bigcup_{C \in Comp}$ {\bf WMEM-Cover-3$((X,{\cal F},w)_C)$}; \\
\> {\bf else}\\
3.2 \>\> find subset $x$ with maximum degree in the intersection graph;\\
3.3 \>\> {\bf if} the degree of $x$ is at most $2$ {\bf then}\\
3.3.1 \>\>\> {\bf return WMEM-Cover-2$((X,{\cal F},w))$};\\
\>\> {\bf else}\\
3.3.2 \>\>\> find the node $x$ such that $x$ is the first node with degree $3$ that is connected\\
\>\>\> ~ to a node with minimum degree in the intersection graph or $x$ is any node  \\
\>\>\> ~ if degrees of all nodes in the intersection graph are $3$;\\
3.3.3 \>\>\> $Solution_1 = \{x\} \cup$ {\bf WMEM-Cover-3$((X-x,{\cal F}-neighbor(x),w))$};\\
3.3.4 \>\>\> $Solution_2 = $ {\bf WMEM-Cover-3$((X,{\cal F}-x,w))$};\\
3.3.5 \>\>\> {\bf return} the best solution among $Solution_1$ and
$Solutoin_2$;

\end{tabbing}
%\end{quote}
\vspace*{-3mm}
%  \strut} $$\boxit{\box4}$$ \vspace*{-15mm}
\caption{ The main algorithm for the {\sc weighted mutually
exclusive maximum set cover} problem with overlapped degrees
bounded by $3$.} \label{Algorithm_3}
\end{figure}

\begin{figure}[ht]
\centering \scalebox{0.7}{\includegraphics{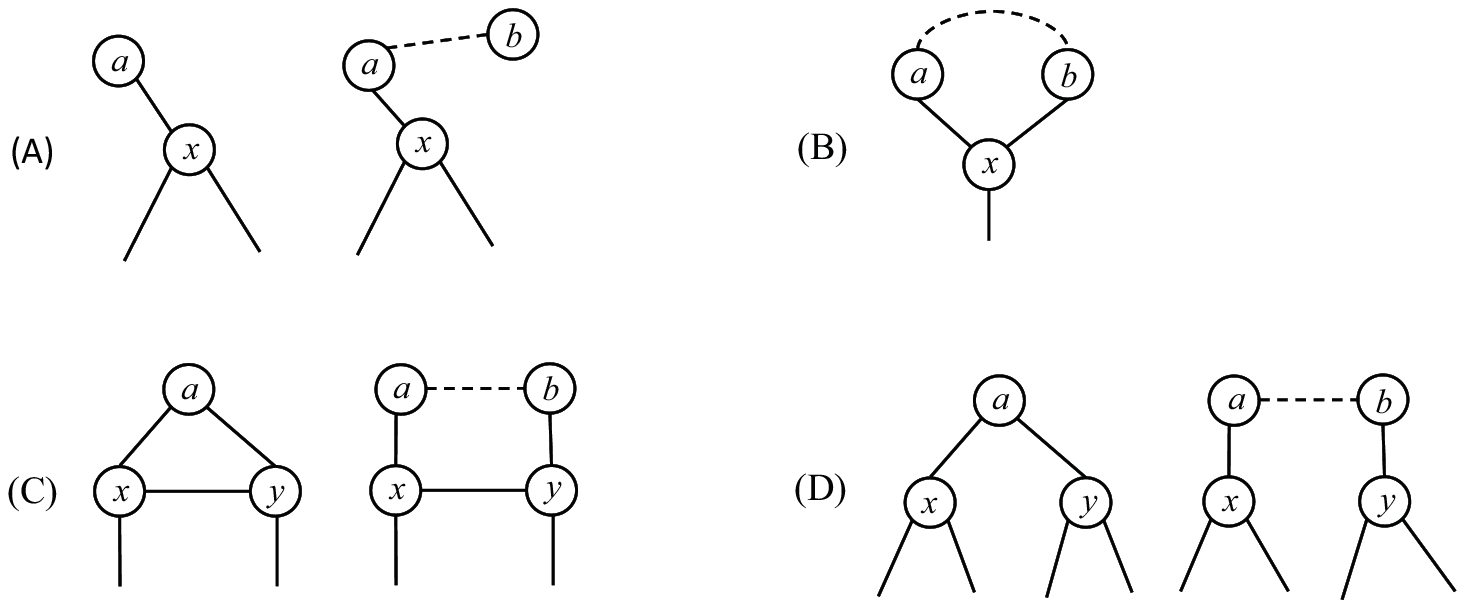}}
\caption{Different structures in the intersection graph with
degree bounded by $3$.} \label{fig_structure}
\end{figure}

In next lemma, we will present how to improve the running time of
algorithm when the degrees of nodes in the intersection graph is
bounded by $3$.

\begin{lemma}\label{lemma3}
Given an instance $(X, {\cal F}, w)$ of the {\sc weighted mutually
exclusive maximum set cover} problem, if the degree of its
intersection graph is bounded by $3$, then the problem can be
solved in $O^*(1.325^m)$ time.
\end{lemma}
\begin{proof}
We suppose that the intersection graph always has a node whose
degree is less than $3$. At the beginning, if the degrees of all
nodes in the intersection graph are $3$, then after the first
branching, both subgraphs will have at least $3$ nodes whose
degrees are at most $2$. After that, when the algorithm makes new
branchings, it is obvious that there are always new nodes whose
degrees will be reduced. Hence, after the first branching, the
intersection graph will always keeps at least one node of degree
bounded by $2$.

The algorithm {\bf WMEM-Cover-3} always first finds a node $x$ of
degree $3$, which is the first node that is connected to a node
with minimum degree (less than $3$) in the intersection graph,
then branches at $x$. We analysis the running time of the
algorithm {\bf WMEM-Cover-3} by considering the following cases.

{\bf Case 1.} The node $x$ is connected by a simple path $P$ that
one end is not connected to any other node (refer to
Figure~\ref{fig_structure}-(A)). In the branch of including $x$
into the solution, $x$ and $3$ neighbors of $x$ will be removed.
In the branch that excludes $x$ from the solution, $x$ is removed;
the simple path $P$ becomes an isolated component and the
sub-instance induced by $P$ can be solved in polynomial time; thus
at least $2$ nodes will be removed in this branch. We obtain the
recurrence relation
\[T(m)\leq T(m-4) + T(m-2),\] which leads to $T(m)\leq 1.273^m$.

{\bf Case 2.} Both ends of the simple path $P$ are connected to
$x$ (refer to Figure~\ref{fig_structure}-(B)), where in this case,
the length of the simple path $P$ is at least $2$. Then in the
branch of including $x$ into the solution, as the case 1, at least
$4$ nodes will be removed and in the branch of excluding $x$ from
the solution, the path $P$ also becomes an isolated component.
Hence, we will have
\[T(m)\leq T(m-4) + T(m-3),\] which leads to $T(m)\leq 1.221^m$.

{\bf Case 3.} One end of the simple path $P$ is connected to $x$
while the other end of $P$ is connected to node $y$ that is not
$x$, where $x$ and $y$ can be or is not connected by an edge
(refer to Figure~\ref{Algorithm_main}-(C)(D)). In the branch that
includes $x$ into the solution, as the above cases, at least $4$
nodes will be removed. In the other branch, after $x$ is removed,
a node of degree one will be generated. If no node(s) of degree
one is in the connected component with nodes of degree $3$, then
node(s) of degree one is/are in connected component(s) bounded
$2$. Hence, we will have
\[T(m)\leq T(m-4) + T(m-2),\] which leads to $T(m)\leq 1.273^m$.
If there is at least one node of degree one is in the connected
component with nodes of degree $3$, then next branching is as the
Case 1. Therefore, even in the worst case, we will have the
recurrence relation
\[T(m)\leq T(m-4) + T(m-1) \leq T(m-4) + (T(m-5) + T(m-3)),\] which leads to $T(m)\leq 1.325^m$.

Above analysis has included all possible situations that a node of
degree at most $2$ is connected to a node of degree $3$. Hence, we
can obtain that the time complexity of the algorithm is
$O^*(1.325^m)$.

As Lemma~\ref{lemma2}, the correctness of the algorithm {\bf
WMEM-Cover-3} is obvious. \qed
\end{proof}

\begin{figure}[htb]
%\vspace{-3mm}%\setbox4=\vbox{\hsize32pc
%\noindent\strut \begin{quote}
\footnotesize
\begin{tabbing}
xxxx\=xx\=xx\=xx\=xx\=xx\=xx\=xx\=xx\=xx\=xx\=xx\=xx\=\kill
{\bf WMEM-Cover-main$((X,{\cal F},w))$}\\
\textbf{Input:} An instance of the {\sc weighted mutually exclusive maximum set cover} problem. \\
\textbf{Output:} A mutually exclusive maximum set cover with minimum weight.\\

1 \> {\bf if} $X=\emptyset$ or ${\cal F} = \emptyset$ {\bf then}\\
1.1 \>\> {\bf return} $\emptyset$;\\
2 \> Find all connected components in the intersection graph and save them to $Comp$;\\
3.\> {\bf if} the number of components is larger than $2$ {\bf then} \\
3.1 \>\>\ {\bf return} $\bigcup_{C \in Comp}$ {\bf WMEM-Cover-main$((X,{\cal F},w)_C)$}; \\
\> {\bf else}\\
3.2 \>\> find subset $S$ with maximum degree in the intersection graph;\\
3.3 \>\> {\bf if} the degree of $S$ is at most $3$ {\bf then}\\
3.3.1 \>\>\> {\bf return WMEM-Cover-3$((X,{\cal F},w))$};\\
\>\> {\bf else}\\
3.3.2 \>\>\> find node $x$ that has the maximum degree in the
intersection graph;\\
3.3.3 \>\>\> $Solution_1 = \{x\} \cup$ {\bf WMEM-Cover-main$((X-x,{\cal F}-neighbor(x),w))$;}\\
3.3.4 \>\>\> $Solution_2 = $ {\bf WMEM-Cover-main$((X,{\cal F}-x,w))$};\\
3.3.5 \>\>\> {\bf return} the best solution among $Solution_1$ and
$Solutoin_2$;

\end{tabbing}
%\end{quote}
\vspace*{-3mm}
%  \strut} $$\boxit{\box4}$$ \vspace*{-15mm}
\caption{ The main algorithm for the {\sc weighted mutually
exclusive maximum set cover} problem.} \label{Algorithm_main}
\end{figure}

Next, we will present the main result.

\begin{theorem}\label{theorem2}
The {\sc weighted mutually exclusive maximum set cover} problem
can be solved in $O^*(1.325^m)$ time.
\end{theorem}
\begin{proof}
As above lemmas, the correctness of the algorithm {\bf
WMEM-Cover-main} is trivial. We only prove the running time of the
algorithm.

The algorithm {\bf WMEM-Cover-main} always keeps searching the
node $x$ with maximum degree in the intersection graph. Then
branches on it. If the degree of $x$ is $d$, then we will obtain
the recurrence relation
\[T(m) \leq T(m-(d+1)) + T(m-1).\]
Furthermore, if $d\leq 3$, the algorithm {\bf WMEM-Cover-main}
will call the algorithm {\bf WMEM-Cover-3}. Hence $d \geq 4$ for
the branching in the algorithm of {\bf WMEM-Cover-main}, which
lead to $T(m)\leq 1.325^m$. From Lemma~\ref{lemma3}, if $d\leq 3$,
we also have $T(m)\leq 1.325^m$. Therefore, the overall running
time of the algorithm {\bf WMEM-Cover-main} is $O^*(1.325^m)$.
 \qed
\end{proof}

\section{Conclusion and future works}
In this paper, first we proved that the {\sc weighted mutually
exclusive maximum set cover} problem is NP-hard. Then we designed
the first non-trivial algorithm that uses $m$, the number of
subsets in ${\cal F}$ as parameter, for the problem. In our
algorithm, we created an intersection graph that can easily help
us to find branching subsets that can greatly reduce the time
complexity of the algorithm. The running time of our algorithm is
$O^*(1.325^m)$, which can easily finish the computation in the
application of finding driver mutations if $m$ is less than $100$.

By choosing the branching subsets more carefully, we believe that
the running time of the algorithm can be further improved. While
reducing the running time to solve the {\sc weighted mutually
exclusive maximum set cover} problem, which has important
applications in cancer study, is appreciated, another variant of
the problem should be particularly paid attention to. Strict
mutual exclusivity is the extreme case, some tumors may have more
than one perturbation to the pathway. Hence, we need to relax the
strict mutual exclusivity and modify the problem to the {\sc
weighted small overlapped maximum set cover problem}, which allow
each tumor to be covered by a small number (such as 2 or 3) of
mutations. This is another important problem, which is abstracted
from the cancer study, need to be solved.

%\bibliography{ExclusiveCover}

\end{document}